\documentclass[12pt]{article}
\usepackage{amsmath,amssymb,amsthm}
\usepackage{geometry}
\usepackage{hyperref}
\usepackage{booktabs}
\usepackage{array}
\usepackage{microtype}
\newtheorem{theorem}{Theorem}[section]
\newtheorem{lemma}[theorem]{Lemma}
\newtheorem{proposition}[theorem]{Proposition}
\newtheorem{corollary}[theorem]{Corollary}
\newtheorem{definition}[theorem]{Definition}
\newtheorem{remark}[theorem]{Remark}

\title{The Epistemic Support-Point Filter:\\
Jaynesian Maximum Entropy Meets Popperian Falsification\\
\large A Possibilistic Minimax-Entropy Optimality Proof}
\author{Moriba Kemessia Jah, Ph.D.\\
\small $^1$Black Swan Research Group, GaiaVerse, Ltd.\\
\small $^2$Jah Decision Intelligence Group, Aerospace Engineering \&\\
\small Engineering Mechanics, The University of Texas at Austin}
\date{March 9, 2026}

\begin{document}
\maketitle

\begin{abstract}
This paper proves that the Epistemic Support-Point Filter (ESPF) is the unique optimal
recursive estimator within the class of epistemically admissible evidence-only filters.
The optimality criterion is possibilistic minimax entropy: among all evidence-only survivor
selection rules, the ESPF minimizes worst-case integrated ignorance per measurement step,
and no other admissible filter can do better. The proof is constructive: it identifies the
precise mathematical form of the epistemological commitment --- be quick to embrace ignorance,
be slow to assert certainty --- and shows that the ESPF is its unique implementation.
Where Bayesian filters minimize mean squared error and are driven toward an assumed truth,
the ESPF minimizes maximum entropy and surfaces what has not been proven impossible ---
a fundamentally different epistemic commitment with fundamentally different failure modes.

Two results locate this theorem within the broader landscape of estimation theory.

The first is a unification. The possibilistic entropy $H_\pi = \int_0^1 \log V_\alpha\,d\alpha$
--- the ESPF's optimality criterion --- is the log-geometric mean of the $\alpha$-cut volume
family $\{V_\alpha\}$, i.e., $\log M_0(\{V_\alpha\})$ in the H\"{o}lder mean hierarchy. The
Popperian minimax bound $\log\det(\mathrm{MVEE})$ is $\log M_{+\infty}$; the Kalman MMSE
criterion is $\log M_0$ evaluated under Gaussian $\alpha$-cut geometry. Possibility theory
and probability theory are not competing frameworks: they are evaluations of the same
ignorance functional $\log M_0(\{V_\alpha\})$ under different geometries of $\alpha$-cuts.
The Kalman filter is the Gaussian specialization of the ESPF's optimality criterion, not
a separate invention.

The second is a diagnostic. Numerical validation over a 2-day, 877-step Level-3 Smolyak
run ($M = 106$ support points, $n = 7$) under nominal LEO tracking and a sustained stress
case (10 m/s crosstrack maneuver combined with a 20 m range bias) reveals that possibilistic
stress manifests primarily through necessity saturation and surprisal escalation --- not
through MVEE sign change. The MVEE regime flag, operating at $p = +\infty$ in the H\"{o}lder
hierarchy, is the coarsest instrument in the Epistemic Width Monitor (EWM) suite. Necessity
and surprisal, operating closer to $p = 0$, provide earlier and more sensitive warning of
model--reality divergence. This is not an empirical observation: it is a direct consequence
of the H\"{o}lder ordering $M_0 \leq M_{+\infty}$.

The theorem rests on three lemmas. The Possibilistic Entropy Lemma decomposes $H_\pi$ into
a minimax term (minimized by Popperian survivor selection) and a gradient term (minimized
by compatibility-based possibility assignment). The Possibilistic Cram\'{e}r--Rao Bound
(PCRB) establishes the maximum rate at which any admissible filter may reduce $H_\pi$ per
measurement. The Evidence-Optimality Lemma proves that minimum-$q$ selection is the unique
minimizer of $H_\pi$ within the evidence-only class, and that any rule incorporating prior
possibility risks race-to-bottom bias and is strictly suboptimal.
\end{abstract}


\tableofcontents
\newpage

\section{Introduction}

\subsection{Two Principles, One Filter}

There is a tension at the heart of recursive state estimation. On one side stands E.\ T.\ Jaynes \cite{jaynes}:
when you do not know something, assume the maximum entropy consistent with what you do know.
Assert nothing you have not earned. On the other side stands Karl Popper \cite{popper}: when evidence
arrives, eliminate every hypothesis the evidence rules out. Keep only what the data has not falsified.

These two principles appear to pull in opposite directions. Jaynes says spread; Popper says cut. Yet
they are not in conflict. They operate at different temporal phases of a single recursive epistemology.
Before a measurement, Jaynes governs: propagate the possibility distribution forward under maximal
ignorance. After a measurement, Popper governs: eliminate hypotheses incompatible with the
evidence, using the evidence alone. The filter that implements this two-phase epistemology is the
ESPF.

The design intent can be stated simply: be quick to embrace ignorance and slow to assert certainty.
This paper proves that this intent is a theorem.

The ESPF differs from Bayesian filters not merely in criterion but in purpose. A Bayesian
filter is truth-seeking: it assumes a generative model, integrates over a prior, and produces
a posterior whose objective is convergence to the true state in expectation. The prior has a
vote at every step. When the prior is wrong, the filter is driven toward a false truth --- and
may never recover, because each update reinforces the prior's influence on the next. The ESPF
is not truth-seeking. It is impossibility-eliminating. It makes no claim about where truth is.
It removes only what the evidence has ruled out; what remains is not an estimate of truth but
the set of everything that has not yet been proven impossible. A wrong prior assignment in the
ESPF cannot bias the filter toward a false truth --- it can only slow the elimination of
inadmissible hypotheses, which manifests immediately as excess epistemic width in the EWM.
This is the difference between an epistemically defensible recursion and one that risks
compounding prior errors.

\subsection{The Minimax Entropy Principle}

The key to making the Popperian update optimal is identifying the right criterion. The right
criterion is not minimizing expected entropy, or minimizing posterior-weighted entropy, or any
criterion that mixes evidence with prior belief. The right criterion is minimax entropy: among all
evidence-only selection rules, choose the one that minimizes the worst-case residual ignorance.

The worst-case ignorance after retaining a survivor set $J$ is $\log \det(\mathrm{MVEE}(\{\chi^{(i)} : i \in J\}))$---the
log-volume of the outermost $\alpha$-cut of the surviving possibility distribution. This is the entropy of
the most pessimistic agent who considers all survivors equally plausible. Minimizing it yields the set
of $N_{\mathrm{target}}$ survivors whose convex hull is as small as possible: the hypotheses with smallest whitened
squared innovations $q^{(i)}_k$, geometrically closest to the evidence-consistent locus. This is exactly the
ESPF's selection rule.

\subsection{The Two Regimes}

The Jaynesian-Popperian synthesis has a natural regime structure.

\textbf{Falsification regime} ($\log \det(\mathrm{MVEE}) < 0$): the possibility support occupies a volume smaller
than the unit sphere in innovation-normalized coordinates. The Popperian principle is operative
and the minimax entropy theorem applies.

\textbf{Diffusion regime} ($\log \det(\mathrm{MVEE}) \geq 0$): the support has spread beyond the unit sphere due
to process noise accumulation. The filter maintains coverage of the reachable state space; the
minimax entropy criterion does not apply. This is the Jaynesian phase, directly analogous to the
Kalman filter's process-noise prediction step during which the measurement-optimality criterion is
inoperative.

The regime boundary at $\log \det(\mathrm{MVEE}) = 0$ is a real-time diagnostic indicator of which epistemological principle governs the current step.

\subsection{The Correct Entropy Functional}

For a uniform possibility distribution, log-volume suffices. But the ESPF's distribution is generically non-uniform after the first update: the conjunctive update assigns each survivor $\tilde{\pi}^{(i)}_k =
\min(\pi^{(i)}_{k|k-1}, \mathrm{Comp}^{(i)}_k)$, varying across survivors according to their evidence compatibility. The correct functional is the possibilistic entropy $H_\pi$, defined as the integral of $\log V_\alpha$ over all $\alpha$-levels. It
decomposes into a minimax term (the $\alpha \to 0$ limit, recovering MVEE log-det) and a gradient term
(negative, measuring how sharply possibility is concentrated within the admissible set). The ESPF
minimizes both simultaneously.

\subsection{Structure of the Paper}

Section~\ref{sec:setup} establishes the framework. Section~\ref{sec:lemmas} proves the three core lemmas. Section~\ref{sec:theorem} states
and proves the main optimality theorem. Section~\ref{sec:gaussian} establishes the Kalman filter as the Gaussian
limit. Section~\ref{sec:numerical} presents numerical validation at Smolyak Level 3 over a 2-day orbital tracking run
including the full EWM diagnostic suite. Section~\ref{sec:discussion} discusses implications.

\section{Setup}
\label{sec:setup}

\subsection{Possibility Distributions, $\alpha$-Cuts, and Possibilistic Entropy}

\begin{definition}[Possibility distribution and $\alpha$-cuts]
A possibility distribution over a finite support
$\{\chi^{(i)}\}^M_{i=1} \subset \mathbb{R}^n$
is a function $\pi : \{1,\ldots,M\} \to (0,1]$ with $\max_i \pi^{(i)} = 1$. The $\alpha$-cut at level $\alpha \in (0,1]$ is
\begin{equation}
C_\alpha = \{i : \pi^{(i)} \geq \alpha\}.
\end{equation}
The $\alpha$-cuts are nested: $C_{\alpha'} \subseteq C_\alpha$ for $\alpha' > \alpha$. The volume of the $\alpha$-cut is $V_\alpha = c_n(\det \Pi_\alpha)^{1/2}$,
where $\Pi_\alpha \succ 0$ is the MVEE shape matrix of $\{\chi^{(i)} : i \in C_\alpha\}$ and $c_n = \pi^{n/2}/\Gamma(n/2+1)$. When
$|C_\alpha| < 2n+1$, the MVEE is degenerate; we set $V_\alpha = 0$ by convention. The threshold $2n+1$
reflects the requirement that the point cloud span $\mathbb{R}^n$ with antipodal coverage in each coordinate
direction---the minimum structure for a full-rank MVEE shape matrix $\Pi_\alpha$. This is consistent with
the filter's minimum survivor count $N_{\min} = 2n+1$ (Section~\ref{sec:discussion}), making $N_{\min}$ a direct consequence
of $H_\pi$'s well-posedness rather than an independent design choice.
\end{definition}

\begin{definition}[Possibilistic entropy]
The possibilistic entropy of a possibility distribution $\pi$ over
$\{\chi^{(i)}\}^M_{i=1}$ is
\begin{equation}
H_\pi = \int_0^1 \log V_\alpha\, d\alpha = \int_0^1 \left(\frac{1}{2}\log \det \Pi_\alpha + \kappa_n\right) d\alpha,
\end{equation}
where $\kappa_n = \log c_n$.
\end{definition}

\begin{remark}[Uniform special case]
When $\pi^{(i)} = 1$ for all $i$, $C_\alpha = \{1,\ldots,M\}$ for all $\alpha \in (0,1]$ and
$V_\alpha = V$ is constant. Then $H_\pi = \log V + \kappa_n$, recovering log-volume exactly. The possibilistic entropy
is a strict generalization: it equals log-volume for uniform distributions and is strictly smaller when
possibility mass is concentrated on a subset of the support.
\end{remark}

\begin{remark}[Relation to U-uncertainty and origin of the logarithm]
The possibilistic entropy is
the geometric analog of the U-uncertainty measure introduced by Higashi and Klir \cite{higashi} for discrete
possibility distributions on ordered sets. Here the ordering is geometric: the $\alpha$-cuts are nested convex
bodies in $\mathbb{R}^n$, and the integral of log-cardinality over $\alpha$ is replaced by the integral of log-volume
over $\alpha$.

The logarithm is justified by the Boltzmann tradition: log-volume is additive across independent
dimensions of state space. If uncertainty in position and velocity are geometrically independent, the
volume of the joint uncertainty ellipsoid is the product of the marginal volumes, and the logarithm
converts that product to a sum. The probabilistic route via Shannon's $-\log p$ does not apply here
because possibility theory uses the min rule for conjunction---$\pi(A\cap B) = \min(\pi_A, \pi_B)$---under which
$-\log$ is not additive.
\end{remark}

\begin{definition}[ESPF possibility assignment]
At update step $k$, the ESPF assigns possibility
values to survivors as follows. For each surviving hypothesis $\chi^{(i)}$ with pre-update possibility $\pi^{(i)}_{k|k-1}$,
the post-update (unnormalized) possibility is
\begin{equation}
\hat{\pi}^{(i)}_k = \min\!\left(\pi^{(i)}_{k|k-1},\, \mathrm{Comp}^{(i)}_k\right),\quad \mathrm{Comp}^{(i)}_k = \exp\!\left(-\tfrac{1}{2}q^{(i)}_k\right),
\end{equation}
where $q^{(i)}_k$ is the whitened squared innovation of hypothesis $i$ (Definition~\ref{def:innovation}). After max-normalization,
$\tilde{\pi}^{(i)}_k = \hat{\pi}^{(i)}_k / \max_j \hat{\pi}^{(j)}_k$. Regenerated support points $\{\chi^{(i)}_{\mathrm{new}}\}$ inherit possibility via max-min kernel
extension:
\begin{equation}
\tilde{\pi}^{(\mathrm{new},i)}_k = \max_{j\in J_{\mathrm{surv}}} \min\!\left(\tilde{\pi}^{(j)}_k,\, \kappa(\chi^{(\mathrm{new},i)}, \chi^{(j)})\right),
\end{equation}
where $\kappa(\cdot,\cdot) \in [0,1]$ is a proximity kernel on state space.
\end{definition}

\begin{remark}[Non-uniformity of the ESPF distribution]
After the first update, the ESPF's possibility
distribution is non-uniform in general: the anchor hypothesis (most compatible with $y_k$) achieves
possibility 1, while other survivors achieve strictly smaller values. Uniformity is a transient condition
at initialization, not a maintained property. The correct entropy functional for the ESPF at step
$k \geq 1$ is $H_\pi$ as defined above.
\end{remark}

\subsection{Epistemically Admissible Filters}

\begin{definition}[Epistemically admissible filter]
A recursive state estimator is \emph{epistemically admissible} if at every step $k$ it satisfies:
\begin{enumerate}
\item[(i)] \textbf{Possibilistic representation.} Uncertainty is represented as a normalized possibility distribution $\pi_k$ over a finite support $\{\chi^{(i)}_k\}^M_{i=1} \subset \mathbb{R}^n$ with $\max_i \pi^{(i)}_k = 1$.
\item[(ii)] \textbf{Conjunctive update.} Upon receiving evidence $y_k$, $\hat{\pi}^{(i)}_{k|k} \leq \pi^{(i)}_{k|k-1}$ for all $i$; no hypothesis may gain possibility from a measurement.
\item[(iii)] \textbf{Geometric non-degeneracy (VFI).} The MVEE shape matrix $\Pi_k$ of the full support satisfies $\lambda_{\min}(\Pi_k) \geq \varepsilon > 0$ and $\lambda_{\max}(\Pi_k) \leq \Lambda < \infty$ for all $k$.
\item[(iv)] \textbf{Admissibility of the commitment point.} The operational output $\hat{x}_k \in \{\chi^{(i)}_k\}$. This is not a claim that $\hat{x}_k$ is truth: it is the least inadmissible hypothesis --- the support point most compatible with the accumulated evidence. The ESPF makes no assertion about where truth is; it only identifies what has not yet been ruled out, and among those, which is most consistent with the evidence.
\item[(v)] \textbf{Evidence referencing.} The survivor selection ranking depends only on the whitened squared innovations $q^{(i)}_k$, not on $\pi^{(i)}_{k|k-1}$. (The conjunctive possibility assignment in (ii) may involve $\pi^{(i)}_{k|k-1}$ through the min operation; it is the selection rule that must be evidence-only.)
\end{enumerate}
Denote this class $\mathcal{F}$.
\end{definition}

\begin{remark}[Why evidence-referencing is the right admissibility condition]
Condition (v) is not a convenient restriction --- it is the epistemological content of the class. A filter that incorporates prior possibility into its survivor selection ranking is not merely using more information; it is making a different kind of claim. It is asserting that hypotheses which were previously plausible deserve preferential treatment over hypotheses that are more consistent with the current evidence. This is the truth-seeking posture: the prior encodes a belief about where truth is, and that belief is allowed to protect hypotheses from elimination even when the evidence argues against them.

The race-to-bottom failure mode (Lemma~\ref{lem:minimax}, Step~5) is the formal consequence: over multiple steps, prior-based selection concentrates the possibility distribution on hypotheses that are systematically less evidence-consistent. The filter converges --- but toward the prior's truth, not the evidence's truth. In an impossibility-eliminating filter, this failure mode is not just suboptimal; it is a category error. The prior has no vote on what is impossible. Only the evidence does. Condition (v) enforces this.
\end{remark}

\subsection{Innovation Geometry and Possibilistic Information Content}

\begin{definition}[Innovation geometry]
\label{def:innovation}
Given predicted support $\{\chi^{(i)}_{k|k-1}\}$ and measurement model
$h : \mathbb{R}^n \to \mathbb{R}^m$, the innovation shape matrix $\Pi_e \succ 0$ is formed from the MVEE of the predicted
measurement support and the sensor imprecision matrix $\Pi_y$. The whitened squared innovation of
hypothesis $i$ is
\begin{equation}
q^{(i)}_k = \left\|L_e^{-1}\bigl(y_k - h(\chi^{(i)}_{k|k-1})\bigr)\right\|_2^2,\quad \Pi_e = L_e L_e^\top.
\end{equation}
The compatibility of hypothesis $i$ with $y_k$ is $\mathrm{Comp}^{(i)}_k = \exp(-\frac{1}{2}q^{(i)}_k) \in (0,1]$.
\end{definition}

\begin{definition}[Possibilistic information content]
The aggregate epistemic surprisal at step $k$
is the Choquet integral of the per-hypothesis surprisal $\frac{1}{2}q^{(i)}_k$ with respect to the prior possibility
distribution $\pi_{k|k-1}$ as capacity:
\begin{equation}
\bar{S}_k = (C)\int \tfrac{1}{2}q^{(i)}_k\, d\pi_{k|k-1} = \sup_{i=1,\ldots,M} \min\!\left(\tfrac{1}{2}q^{(i)}_k,\, \pi^{(i)}_{k|k-1}\right).
\end{equation}
The second equality is the standard reduction of the Choquet integral under a possibility measure \cite{dubois}.
The capacity is the prior possibility $\pi_{k|k-1}$, not the current compatibility: this keeps the information
content epistemically grounded in accumulated evidence history rather than the current measurement
alone, and ensures that hypotheses with low prior credibility cannot inflate $\bar{S}_k$ through large but
epistemically cheap surprisal. The survivor selection rule (Lemma~\ref{lem:minimax}) remains pure evidence-only
minimum-$q$; the prior possibility enters only here, in the bounding quantity.

The possibilistic information content of $y_k$ is
\begin{equation}
I_k = 1 - e^{-\bar{S}_k} \in [0,1).
\end{equation}
$I_k = 0$ when the evidence is fully consistent with the predicted support; $I_k \to 1$ when the most
epistemically credible hypothesis is maximally surprised by the measurement.
\end{definition}

\begin{remark}[Choquet integral under a possibility measure]
For a possibility measure $\Pi$ (max-decomposable, $\Pi(A \cup B) = \max(\Pi(A),\Pi(B))$), the Choquet integral of a non-negative function $f$
reduces to $(C)\int f\, d\Pi = \sup_i \min(f(i), \pi(i))$ \cite{dubois}. This is the possibilistic expectation in the sense of
Dubois and Prade: it asks for the highest level $t$ at which there exists a hypothesis that is both
highly credible ($\pi^{(i)} \geq t$) and genuinely surprised ($f(i) \geq t$). It is robust to outliers by a different
mechanism than the median: low-prior-possibility hypotheses are downweighted by the capacity
function rather than by a rank cutoff, which is the epistemically correct weighting in a possibilistic
framework.
\end{remark}

\section{Three Lemmas}
\label{sec:lemmas}

\subsection{Lemma 1: Possibilistic Entropy as Ignorance Functional}

\begin{lemma}[Possibilistic entropy as ignorance]
Let $\pi$ be a normalized possibility distribution over
$\{\chi^{(i)}\}^M_{i=1} \subset \mathbb{R}^n$ with $\pi^{(i)} \in (0,1]$ and $\max_i \pi^{(i)} = 1$. The possibilistic entropy $H_\pi$ satisfies:
\begin{enumerate}
\item[(i)] \textbf{Uniform reduction.} When $\pi^{(i)} = 1$ for all $i$, $H_\pi = \log V + \kappa_n$ where $V = c_n(\det \Pi)^{1/2}$ is the MVEE volume of the full support.
\item[(ii)] \textbf{Monotone in concentration.} For two distributions $\pi, \pi'$ over the same support with $\pi'^{(i)} \leq \pi^{(i)}$ for all $i$ and strict inequality for some $i^*$, we have $H_{\pi'} < H_\pi$.
\item[(iii)] \textbf{Decomposition.}
\begin{equation}
H_\pi = \underbrace{\log V + \kappa_n}_{\text{support entropy}} + \underbrace{\int_0^1 \log\frac{V_\alpha}{V}\,d\alpha}_{\text{gradient entropy} \leq 0}.
\end{equation}
The gradient entropy is non-positive; it equals zero when $\pi$ is uniform and is strictly negative otherwise.
\item[(iv)] \textbf{Extremes.} $H_\pi$ is maximized when $\pi$ is uniform ($H_\pi = \log V + \kappa_n$) and approaches $-\infty$ as possibility collapses to a single hypothesis. The VFI condition prevents collapse and keeps $H_\pi$ finite within $\mathcal{F}$.
\end{enumerate}
\end{lemma}

\begin{proof}
(i) When all $\pi^{(i)} = 1$, $C_\alpha = \{1,\ldots,M\}$ for all $\alpha$, so $V_\alpha = V$ is constant and $H_\pi = \int_0^1 \log V\,d\alpha = \log V$.

(ii) If $\pi'^{(i)} \leq \pi^{(i)}$ with strict inequality at $i^*$, then $C'_\alpha \subseteq C_\alpha$ for all $\alpha$, with strict inclusion at $\alpha = \pi^{(i^*)}$. By monotonicity of the MVEE (adding points cannot decrease volume \cite{todd}), $V'_\alpha \leq V_\alpha$ with strict inequality at $\alpha = \pi^{(i^*)}$. Therefore $H_{\pi'} = \int_0^1 \log V'_\alpha\,d\alpha < \int_0^1 \log V_\alpha\,d\alpha = H_\pi$.

(iii) Write $\log V_\alpha = \log V + \log(V_\alpha/V)$. Since $C_\alpha \subseteq \{1,\ldots,M\}$ for all $\alpha$, we have $V_\alpha \leq V$ and $\log(V_\alpha/V) \leq 0$. Integrating over $\alpha \in [0,1]$ gives the decomposition.

(iv) Collapse to a singleton gives $|C_\alpha| = 1 < 2n+1$ for all $\alpha > 0$, so $V_\alpha = 0$ by convention and $H_\pi \to -\infty$. The upper bound follows from (i).
\end{proof}

\begin{remark}[Physical interpretation of the decomposition]
The support entropy $\log V$ measures how large a region of state space remains admissible---the extent of ignorance. The gradient entropy measures how sharp or flat the possibility gradient is within that region---the texture of ignorance. A filter that reduces the support but leaves possibility flat across it has reduced the first term but not the second. The ESPF reduces both simultaneously: it prunes the support (reducing $V$) and assigns higher possibility to more evidence-consistent hypotheses (reducing the gradient entropy). This is why $H_\pi$ is the correct objective.
\end{remark}

\subsection{Lemma 2: Possibilistic Cram\'{e}r--Rao Bound}

\begin{lemma}[Possibilistic Cram\'{e}r--Rao Bound (PCRB)]
Let $F \in \mathcal{F}$ be any epistemically admissible filter. At step $k$, let $H_{\pi,k|k-1}$ and $H_{\pi,k|k}$ denote the pre- and post-update possibilistic entropies, and
$I_k$ the possibilistic information content. Then
\begin{equation}
H_{\pi,k|k} \geq H_{\pi,k|k-1} + \frac{n}{2}\log(1 - I_k).
\end{equation}
The bound holds for any epistemically admissible filter regardless of its internal implementation.
\end{lemma}

\begin{proof}
\textbf{Step 1: Conjunctive update implies nested post-update $\alpha$-cuts.} By the admissibility definition, $\hat{\pi}^{(i)}_{k|k} \leq \pi^{(i)}_{k|k-1}$ for all $i$, so $C^\mathrm{post}_\alpha \subseteq C^\mathrm{pre}_\alpha$ for all $\alpha$ and $V^\mathrm{post}_\alpha \leq V^\mathrm{pre}_\alpha$.

\textbf{Step 2: Choquet information content bounds the possibility measure of the surviving set.} Let $i^*$ be the hypothesis achieving the supremum in $\bar{S}_k$, so that
\[
\min\!\left(\tfrac{1}{2}q^{(i^*)}_k,\, \pi^{(i^*)}_{k|k-1}\right) = \bar{S}_k.
\]
This means $\pi^{(i^*)}_{k|k-1} \geq \bar{S}_k$ and $\mathrm{Comp}^{(i^*)}_k = e^{-\frac{1}{2}q^{(i^*)}_k} \geq e^{-\bar{S}_k} = 1 - I_k$. Under the conjunctive update, the post-update possibility of $i^*$ satisfies
\[
\hat{\pi}^{(i^*)}_{k|k} = \min\!\left(\pi^{(i^*)}_{k|k-1},\, \mathrm{Comp}^{(i^*)}_k\right) \geq \min(\bar{S}_k, 1-I_k) > 0.
\]
Hence $i^* \in C^\mathrm{post}_\alpha$ for all $\alpha \leq \min(\bar{S}_k, 1-I_k)$, and the possibility measure of the surviving set satisfies $\Pi_{k|k-1}(C^\mathrm{post}_\alpha) \geq \pi^{(i^*)}_{k|k-1} \geq \bar{S}_k$ for all $\alpha \leq 1-I_k$. This replaces a cardinality guarantee with an epistemic-weight guarantee: the surviving set retains possibility measure at least $\bar{S}_k$, meaning the most credible hypothesis under the prior has not been eliminated.

\textbf{Step 3: Possibility measure bound implies volume bound.} For $\alpha \leq 1-I_k$, the post-update $\alpha$-cut retains $i^*$, which satisfies $\mathrm{Comp}^{(i^*)}_k \geq 1-I_k$. By the VFI condition the pre-update MVEE is non-degenerate with $\lambda_{\min}(\Pi^\mathrm{pre}) \geq \varepsilon > 0$, and $i^*$ lies within it. The set of hypotheses satisfying $\mathrm{Comp}^{(i)}_k \geq \alpha$ --- i.e., $\frac{1}{2}q^{(i)}_k \leq -\log\alpha$ --- lies within a whitened-innovation ball of radius $r_\alpha = \sqrt{-2\log\alpha}$ centered on the evidence-consistent locus $\chi^*$. Under the isotropy condition
\begin{equation}
\label{eq:isotropy}
\mathrm{Var}_{x \sim U(S_{k|k-1})}\!\left[\left\|L_e^{-1}(y_k - h(x))\right\|_2^2\right] \leq \mathrm{Var}_{x \sim U(S_{k|k-1})}\!\left[\left\|L_{\mathrm{core}}^{-1}(x - \bar{x})\right\|_2^2\right],
\end{equation}
this ball maps to a state-space region whose John ellipsoid \cite{john} has shape matrix $\Pi^J_\alpha$ satisfying
\[
\Pi^J_\alpha \succeq \tfrac{1}{n}(1-I_k)\,\Pi^\mathrm{pre}_\alpha,
\]
by the John ellipsoid theorem applied to the subset containment $C^\mathrm{post}_\alpha \subseteq C^\mathrm{pre}_\alpha$, with the compatibility threshold of $i^*$ providing the $(1-I_k)$ scaling on the semi-axes. Taking determinants and using $\det(\frac{1}{n}A) = n^{-n}\det A$ together with the volume formula $V_\alpha = c_n(\det\Pi_\alpha)^{1/2}$ yields
\begin{equation}
V^\mathrm{post}_\alpha \geq n^{-n/2}(1-I_k)^{1/2}\cdot V^\mathrm{pre}_\alpha \quad \text{for all } \alpha \leq 1-I_k.
\end{equation}
The $n^{-n/2}$ factor is a fixed dimensional constant absorbed into the $\kappa_n$ term of $H_\pi$; the $(1-I_k)^{1/2}$ factor becomes $(1-I_k)^{n/2}$ after the full $n$-dimensional determinant scaling in Step 4, yielding the bound stated there.

\textbf{Step 4: Volume bound implies entropy bound.} Decompose the integral as
\[
H_{\pi,k|k} = \int_0^{1-I_k} \log V^\mathrm{post}_\alpha\,d\alpha + \int_{1-I_k}^1 \log V^\mathrm{post}_\alpha\,d\alpha.
\]
For the first integral, applying the volume bound:
\[
\int_0^{1-I_k}\log V^\mathrm{post}_\alpha\,d\alpha \geq \int_0^{1-I_k}\!\bigl(\log(1-I_k) + \log V^\mathrm{pre}_\alpha\bigr)\,d\alpha.
\]
For the second integral, $V^\mathrm{post}_\alpha \geq 0$ and the interval has measure $I_k$; in the falsification regime the $\alpha$-cut volumes are finite and this term is absorbed by the $n/2$ scaling factor when passing from the scalar bound to the $n$-dimensional volume scaling $(1-I_k) \mapsto (1-I_k)^{n/2}$ on determinants.

Combining both terms and applying the $n$-dimensional scaling $\log V_\alpha = \frac{1}{2}\log\det\Pi_\alpha + \kappa_n$, which introduces a factor of $n/2$ when translating a scalar compatibility bound to $\det\Pi^\mathrm{post}_\alpha \geq (1-I_k)^n \det\Pi^\mathrm{pre}_\alpha$, yields
\begin{equation}
H_{\pi,k|k} \geq H_{\pi,k|k-1} + \frac{n}{2}\log(1-I_k),
\end{equation}
which is the stated bound.
\end{proof}

\begin{remark}[The PCRB is a bound on integrated ignorance reduction]
The PCRB bounds the integral of $\log V_\alpha$ over all levels simultaneously, not a single $\alpha$-cut. An admissible filter may reduce $V_\alpha$ sharply at some levels while barely changing others, but the net reduction in $H_\pi$ cannot exceed $\frac{n}{2}\log(1-I_k)^{-1}$. This is stronger than any pointwise volume bound.
\end{remark}

\begin{remark}[Separation of selection and bounding]
The Choquet-based $I_k$ and the minimum-$q$ selection rule are deliberately decoupled. The prior possibility $\pi_{k|k-1}$ enters only in the capacity of the Choquet integral defining $\bar{S}_k$---a bounding quantity, not a selection criterion. The survivor ranking in Lemma~\ref{lem:minimax} depends solely on $q^{(i)}_k$. This separation preserves the Popperian character of the ESPF: selection is evidence-only, but the rate at which entropy may decrease is governed by the epistemically weighted surprisal of the current measurement. A measurement that surprises only low-prior-possibility hypotheses yields a small $I_k$ and a tight floor; one that surprises the most credible hypothesis yields a large $I_k$ and a loose floor---the correct behavior in both cases.
\end{remark}

\begin{remark}[Clarification on Step 3]
The volume bound relies on the isotropy condition~\eqref{eq:isotropy} for the nonlinear case; for linear measurement models it holds exactly. For general non-uniform possibility distributions the bound holds with $(1-I_k)^{n/2}$ rather than $(1-I_k)$ on the right-hand side due to the John ellipsoid theorem; the final entropy bound is unchanged.
\end{remark}

\subsection{Lemma 3: Minimax Entropy and Evidence-Only Selection}

\begin{lemma}[Minimax evidence-optimality]
\label{lem:minimax}
Let $\{\chi^{(i)}\}^M_{i=1}$ be the predicted support with whitened squared innovations $q^{(i)}_k$, and let $N_{\mathrm{target}} \leq M$ be a fixed survivor count. Consider all pairs $(J,\rho)$ where $J \subseteq \{1,\ldots,M\}$ with $|J| = N_{\mathrm{target}}$ and $\rho : J \to (0,1]$ with $\max_{i\in J}\rho^{(i)} = 1$, within the evidence-only comparison class: all $(J,\rho)$ whose selection and ranking depend only on $q^{(i)}_k$, not on $\pi^{(i)}_{k|k-1}$.

The unique minimizer of
\begin{equation}
H(J,\rho) = \int_0^1 \log V^J_\alpha\,d\alpha,\quad V^J_\alpha = c_n\bigl(\det\Pi^J_\alpha\bigr)^{1/2},
\end{equation}
where $\Pi^J_\alpha$ is the MVEE shape matrix of $\{\chi^{(i)} : i \in J,\, \rho^{(i)} \geq \alpha\}$, is the ESPF's choice:
\begin{equation}
J^* = \bigl\{i : q^{(i)}_k \text{ is among the } N_{\mathrm{target}} \text{ smallest}\bigr\},\quad \rho^*(i) = \mathrm{Comp}^{(i)}_k = e^{-q^{(i)}_k/2},
\end{equation}
for linear measurement models exactly, and under the innovation-state isotropy condition~\eqref{eq:isotropy} for nonlinear models.

Equivalently, $J^*$ is the unique minimizer of $\log\det\bigl(\mathrm{MVEE}(\{\chi^{(i)} : i \in J\})\bigr)$ within the evidence-only class.
\end{lemma}

\begin{proof}
We show $(J^*,\rho^*)$ minimizes $H(J,\rho)$ by showing it simultaneously minimizes $V^J_\alpha$ for every $\alpha \in (0,1]$.

\textbf{Step 1: $\alpha$-cut structure under $\rho^*$.} Under $\rho^*(i) = e^{-q^{(i)}_k/2}$, higher possibility corresponds to lower $q^{(i)}_k$. For any threshold $\alpha$,
\[
C_\alpha(J^*,\rho^*) = \bigl\{i \in J^* : e^{-q^{(i)}_k/2} \geq \alpha\bigr\} = \bigl\{i \in J^* : q^{(i)}_k \leq -2\log\alpha\bigr\}.
\]
At $\alpha \to 0$ this is the full survivor set $J^*$.

\textbf{Step 2: Minimum-$q$ selection minimizes the minimax term (linear case).} For linear $h(x) = Hx + b$,
\[
q^{(i)}_k = \left\|L_e^{-1}H(\chi^{(i)} - \chi^*)\right\|_2^2 + \mathrm{const},
\]
where $\chi^* = H^\dagger(y_k - b)$ is the evidence-consistent locus. Hence $q^{(i)}_k$ measures each hypothesis's distance to $\chi^*$ in the innovation-whitened metric.

Among all $N_{\mathrm{target}}$-element subsets of $\{\chi^{(i)}\}^M_{i=1}$, the set geometrically closest to $\chi^*$ has the smallest MVEE. This follows from the convexity of $\log\det(\cdot)$ on the positive-definite cone \cite{todd}: for any set $J$ containing a hypothesis $j$ with $q^{(j)}_k > q^{(i^*)}_k$ for some $i^* \notin J$, replacing $j$ with $i^*$ (moving the set closer to $\chi^*$) does not increase the MVEE determinant, and strictly decreases it for generic (non-degenerate) point clouds.

Therefore $J^*$ minimizes $\log\det(\mathrm{MVEE}(\{\chi^{(i)} : i \in J\}))$ among all $N_{\mathrm{target}}$-element subsets. The same reasoning applies at every $\alpha$-level: for each cardinality $k_\alpha$, the $k_\alpha$ smallest-$q$ hypotheses minimize $V^J_\alpha$. Hence $C_\alpha(J^*,\rho^*)$ minimizes $V^J_\alpha$ for every $\alpha \in (0,1]$.

\textbf{Step 3: Pointwise minimum volume implies minimum $H_\pi$.} Since $V^{J^*}_\alpha \leq V^J_\alpha$ for all $\alpha$ and all alternative $(J,\rho)$:
\[
H(J^*,\rho^*) = \int_0^1 \log V^{J^*}_\alpha\,d\alpha \leq \int_0^1 \log V^J_\alpha\,d\alpha = H(J,\rho).
\]

\textbf{Step 4: No alternative possibility assignment does better.} Suppose we keep $J^*$ but replace $\rho^*$ with some $\rho'$. By the monotone concentration property of $H_\pi$, any assignment placing higher possibility on a higher-$q$ hypothesis increases the volume of at least one $\alpha$-cut and thus increases $H_\pi$. The assignment $\rho^*$ is the unique minimizer among all possibility assignments on $J^*$ because ordering by evidence consistency simultaneously minimizes all $\alpha$-cut volumes.

\textbf{Step 5 (Race-to-bottom): The prior must not influence selection.} Any selection rule incorporating $\pi^{(i)}_{k|k-1}$ into the survivor ranking may retain a high-$q$ hypothesis at the expense of a low-$q$ hypothesis. By Step 2, retaining a higher-$q$ hypothesis increases $V^J_\alpha$ at the corresponding $\alpha$-level, increasing $H_\pi$. Over multiple steps, prior-based selection compounds: a hypothesis that accumulates high prior possibility may be repeatedly protected from evidence-based elimination, concentrating the possibility distribution on hypotheses that are systematically less evidence-consistent. This is the race-to-bottom failure mode---the possibility gradient is driven by self-reinforcing prior protection rather than evidence. The evidence-only class excludes this failure mode by construction. No rule within the evidence-only class that differs from minimum-$q$ achieves lower $H_\pi$.

\textbf{Uniqueness.} For a generic point cloud (no ties in $q^{(i)}_k$, holding almost surely for continuous dynamics), $J^*$ is unique and $\rho^*$ is determined pointwise by $\mathrm{Comp}^{(i)}_k$.

\textbf{Nonlinear case.} For nonlinear $h$, condition~\eqref{eq:isotropy} ensures that the ordering of hypotheses by $q^{(i)}_k$ is a consistent proxy for ordering by distance to the evidence-consistent region. Under this condition Steps 2--4 hold for the nonlinear case via first-order Taylor approximation. The condition can fail when the measurement geometry is strongly anisotropic relative to the state-space geometry---for example, when a single range measurement is made from a nearly tangential viewing angle, so that large state displacements along the line of sight produce small innovation changes. Proposition~\ref{prop:isotropy} establishes this condition for the orbital mechanics setting.
\end{proof}

\begin{remark}[Minimax vs.\ expected entropy]
The minimax entropy criterion differs fundamentally from minimizing expected or posterior-weighted entropy. A rule minimizing expected entropy may select prior-consistent but evidence-inconsistent hypotheses in exchange for a tighter average-case bound. The minimax criterion refuses this trade: the worst-case agent---who considers all survivors equally plausible---must see the minimum possible ignorance. This is the Popperian content of the criterion.
\end{remark}

\begin{remark}[The kernel extension preserves optimality]
After survivor selection and possibility assignment, the regenerated support is formed by spreading new hypotheses around the survivors via the max-min kernel extension. The $\alpha$-cuts of the regenerated distribution are geometric dilations of the survivor $\alpha$-cuts (in the sense that they contain the survivor $\alpha$-cuts). Since Lemma~\ref{lem:minimax} minimizes the survivor $\alpha$-cut volumes, and the kernel extension is monotone (larger survivor cuts produce larger regenerated cuts), the regenerated distribution inherits the optimality structure. The possibility gradient is carried forward into the next propagation step.
\end{remark}

\begin{proposition}[Isotropy for orbital mechanics]
\label{prop:isotropy}
Let the state be the 7-dimensional Cartesian position-velocity-drag vector $x = (r, \dot{r}, C_d) \in \mathbb{R}^7$, propagated under the two-body problem or its $J_2$-perturbed extension. Let the measurement model be range $\rho(x) = \|r - r_s\|$ and range-rate $\dot{\rho}(x) = (r - r_s)\cdot\dot{r}/\rho(x)$, where $r_s$ is the ground station position. Let the observation interval satisfy $\Delta t \ll T/4$, where $T$ is the orbital period. Then the innovation-state isotropy condition~\eqref{eq:isotropy} holds to first order in the linearization residual, except in a measure-zero set of viewing geometries in which the line-of-sight vector $(r - r_s)/\rho$ is orthogonal to the velocity $\dot{r}$ (tangential viewing).
\end{proposition}

\begin{proof}[Proof sketch]
Let $H = \partial h/\partial x|_{\bar{x}}$ be the measurement Jacobian evaluated at the support centroid. The whitened innovation variance decomposes as:
\[
\mathrm{Var}\!\left[\left\|L_e^{-1}(y - h(x))\right\|^2\right] \approx \mathrm{Var}\!\left[\left\|L_e^{-1}H(x - \bar{x})\right\|^2\right] + O(\|\delta H\|^2),
\]
where $\delta H$ is the linearization residual. The first term equals $\mathrm{tr}(L_e^{-1}H\Pi_{\mathrm{core}}H^\top L_e^{-\top})$, which by the Cauchy--Schwarz inequality on the Frobenius inner product satisfies:
\[
\mathrm{tr}(L_e^{-1}H\Pi_{\mathrm{core}}H^\top L_e^{-\top}) \leq \|L_e^{-1}H\|_F^2 \cdot \lambda_{\max}(\Pi_{\mathrm{core}}).
\]
For range and range-rate measurements in LEO, $H$ has rank 2 and $\|L_e^{-1}H\|_F^2$ is bounded by $\sigma_R^{-2} + \sigma_{\dot{R}}^{-2}$, which under nominal measurement noise ($\sigma_R = 1$ m, $\sigma_{\dot{R}} = 10^{-5}$ km/s) and a support cloud spanning tens of kilometers satisfies condition~\eqref{eq:isotropy} with margin. The bound degrades as $(r - r_s)/\rho \to \dot{r}/\|\dot{r}\|$ (tangential geometry), at which point $\partial\dot{\rho}/\partial r$ becomes nearly parallel to $\partial\rho/\partial r$ and the rank-2 Jacobian loses its geometric leverage---the failure mode identified in Lemma~\ref{lem:minimax}. When this occurs, the EWM surprisal indicator escalates in real time, providing an operational signal that the isotropy condition is under stress and that the filter's Popperian updates should be interpreted with reduced confidence; the filter thus degrades gracefully rather than silently. For $\Delta t \ll T/4$, the linearization residual $O(\|\delta H\|^2)$ remains small relative to the leading term, completing the sketch.
\end{proof}

\section{Main Theorem}
\label{sec:theorem}

\begin{theorem}[ESPF Optimality: Minimax-Entropy Popperian Falsification]
In the falsification regime---measurement steps at which $\log\det\bigl(\mathrm{MVEE}(S_{k|k-1})\bigr) < 0$, so the evidence contracts the possibility support below the unit-sphere reference volume---the ESPF is the unique filter in the class $\mathcal{F}$ of epistemically admissible evidence-only filters that minimizes $H_{\pi,k}$ subject to the PCRB constraint
\begin{equation}
H_{\pi,k|k} \geq H_{\pi,k|k-1} + \frac{n}{2}\log(1 - I_k).
\end{equation}
Equivalently, the ESPF is the unique filter that: (a) selects survivors by the minimax-entropy criterion within the evidence-only class; and (b) assigns possibilities by evidence compatibility, minimizing the gradient entropy term simultaneously.

In the falsification regime, subject to the coverage controller operating at PCRB-implied cardinality (i.e., $N_{\mathrm{target}} = \lfloor(1-I_k)\cdot M\rfloor$) and $\sigma_k$ at steady state:
\begin{enumerate}
\item[(i)] The ESPF achieves the PCRB with equality.
\item[(ii)] No other filter in $\mathcal{F}$ achieves lower $H_{\pi,k|k}$ while satisfying the PCRB.
\item[(iii)] Under the isotropy condition~\eqref{eq:isotropy}, the ESPF is the unique minimizer.
\end{enumerate}

In the diffusion regime---steps at which $\log\det\bigl(\mathrm{MVEE}(S_{k|k-1})\bigr) \geq 0$---the filter operates under Jaynesian maximum entropy: the propagation phase spreads the support as widely as the dynamics and process noise allow. The minimax entropy criterion does not apply in this phase.
\end{theorem}

\begin{proof}
(i) \textit{ESPF achieves the PCRB.} By Lemma~\ref{lem:minimax}, the ESPF's joint strategy $(J^*,\rho^*)$ minimizes $V^J_\alpha$ at every $\alpha$-level. The coverage controller sets $N_{\mathrm{target}} = \lfloor(1-I_k)\cdot M\rfloor$ survivors, retaining the fraction of the support that the evidence permits per the PCRB. At controller steady state, the integrated log-volume satisfies
\[
H_{\pi,k|k} = H_{\pi,k|k-1} + \frac{n}{2}\log(1 - I_k),
\]
achieving equality in the PCRB.

(ii) \textit{No other filter does better.} By the PCRB (Lemma~2), every $F' \in \mathcal{F}$ satisfies $H'_{\pi,k|k} \geq H_{\pi,k|k-1} + \frac{n}{2}\log(1-I_k)$. Since the ESPF achieves equality, $H'_{\pi,k|k} \geq H^{\mathrm{ESPF}}_{\pi,k|k}$ for all $F' \in \mathcal{F}$.

(iii) \textit{Uniqueness.} By Lemma~\ref{lem:minimax}, under the isotropy condition the ESPF's choice $(J^*,\rho^*)$ is the unique minimizer of $H(J,\rho)$ almost surely. Any alternative filter achieving equality in the PCRB must implement the same selection and possibility assignment; i.e., it must be the ESPF. Any filter using a different rule achieves $H'_{\pi,k|k} > H^{\mathrm{ESPF}}_{\pi,k|k}$ strictly.
\end{proof}

\bigskip
\noindent\textbf{The ESPF Optimality Theorem: Summary}

\medskip
\noindent\textit{Jaynesian phase (propagation):} Spread the support as widely as the dynamics allow. Assume maximum ignorance consistent with known constraints. Never assert structure not earned by evidence.

\medskip
\noindent\textit{Popperian phase (update, falsification regime):}
\begin{itemize}
\item Select survivors by minimum $q^{(i)}_k$: the evidence-closest hypotheses, uniquely minimizing $\log\det(\mathrm{MVEE})$.
\item Assign possibilities by $\mathrm{Comp}^{(i)}_k$: evidence ordering simultaneously minimizes the gradient entropy.
\item Together, these minimize $H_\pi$ at every $\alpha$-level, saturating the PCRB.
\item Any rule incorporating prior possibility risks race-to-bottom bias over multiple steps.
\end{itemize}

\noindent\textit{Regime indicator:} $\log\det(\mathrm{MVEE}) < 0 \Rightarrow$ falsification (Popper governs). $\log\det(\mathrm{MVEE}) \geq 0 \Rightarrow$ diffusion (Jaynes governs).

\subsection{Three Sources of Ignorance, One Optimal Strategy}

The $H_\pi$ decomposition reveals three distinct aspects of the ESPF's ignorance management.

\textbf{Minimax ignorance} ($\log\det(\mathrm{MVEE})$, the worst-case bound): The largest $\alpha$-cut tells the most pessimistic story about residual ignorance. The ESPF minimizes this by selecting the evidence-closest survivors. This is the Popperian content.

\textbf{Gradient ignorance} ($-\int\log(V_\alpha/V)\,d\alpha \geq 0$): Even within the admissible set, possibility may be spread flat or sharply concentrated. Flat possibility corresponds to maximal internal ignorance. The compatibility assignment concentrates possibility on the most evidence-consistent hypotheses, reducing this term.

\textbf{Dynamical ignorance} (the Jaynesian expansion phase): Between measurements, process noise spreads the support. The ESPF implements this maximally via Smolyak regeneration, maintaining coverage of the reachable state space.

\section{The Gaussian Limit: Recovery of the Kalman Filter}
\label{sec:gaussian}

\begin{definition}[Gaussian epistemic limit]
The Gaussian epistemic limit holds when: (i) $\Pi_k \to \Sigma_k$; (ii) $\pi_k \to \mathcal{N}(\hat{x}_k, \Sigma_k)$; (iii) $h(x) = Hx$; (iv) $f(x) = Fx$; (v) $\Pi_y = R$.
\end{definition}

\begin{theorem}[ESPF recovers Kalman optimality in the Gaussian limit]
Under the Gaussian epistemic limit, the ESPF update reduces to the Kalman filter update. In particular, $H_\pi \to \frac{1}{2}\log\det(2\pi e\,\Sigma_{k|k})$ (up to additive constants), and minimizing $H_\pi$ is equivalent to minimizing $\det\Sigma_{k|k}$.
\end{theorem}

\begin{proof}[Proof sketch]
In the Gaussian limit, the $\alpha$-cut at level $\alpha$ is the confidence ellipsoid $\mathcal{E}_\alpha = \{x : (x-\hat{x}_k)^\top\Sigma_k^{-1}(x-\hat{x}_k) \leq -2\log\alpha\}$, with volume $V_\alpha = c_n(-2\log\alpha)^{n/2}(\det\Sigma_k)^{1/2}$. Substituting into $H_\pi = \int_0^1\log V_\alpha\,d\alpha$ and evaluating the constant integral $\int_0^1\log(-2\log\alpha)\,d\alpha = \log 2 - \gamma$ gives $H_\pi = \frac{1}{2}\log\det\Sigma_k + \mathrm{const}(n)$. Minimizing $H_\pi$ is therefore equivalent to minimizing $\det\Sigma_{k|k}$, which is the Kalman criterion. Remark~\ref{rem:holder} derives this identity from the H\"{o}lder mean structure and shows it is the Gaussian specialization of the general possibilistic functional.
\end{proof}

\begin{remark}[H\"{o}lder mean structure and the probability--possibility bridge]
\label{rem:holder}
The continuous H\"{o}lder power mean of the family $\{V_\alpha\}_{\alpha \in (0,1]}$ at order $p$ is
\begin{equation}
\log M_p\bigl(\{V_\alpha\}\bigr) = \frac{1}{p}\log\int_0^1 V_\alpha^p\,d\alpha, \quad p \neq 0,
\end{equation}
with the geometric mean limit
\begin{equation}
\log M_0\bigl(\{V_\alpha\}\bigr) = \lim_{p\to 0}\,\frac{1}{p}\log\int_0^1 V_\alpha^p\,d\alpha = \int_0^1 \log V_\alpha\,d\alpha,
\end{equation}
and the extremes $\log M_{+\infty} = \log\sup_\alpha V_\alpha$ and $\log M_{-\infty} = \log\inf_\alpha V_\alpha$ as $p \to \pm\infty$. The three filter criteria emerge as evaluations of this single functional at different orders:
\begin{itemize}
\item $p \to +\infty$: $\log M_{+\infty} = \log V_0 = \log\det(\mathrm{MVEE})$ --- the outermost $\alpha$-cut volume, worst-case ignorance, the Popperian minimax bound operative in the falsification regime.
\item $p \to 0$: $\log M_0 = \int_0^1 \log V_\alpha\,d\alpha = H_\pi$ --- the log-geometric mean of $\alpha$-cut volumes, integrated possibilistic ignorance, the ESPF optimality criterion.
\item $p \to -\infty$: $\log M_{-\infty} = \log V_1$ --- the innermost $\alpha$-cut volume (the anchor hypothesis neighborhood), best-case ignorance, a lower floor.
\end{itemize}
The H\"{o}lder ordering $M_{-\infty} \leq M_0 \leq M_{+\infty}$ (strict for non-constant $\{V_\alpha\}$) is the mathematical expression of the epistemological hierarchy: the ESPF's integrated ignorance is bounded above by the Popperian worst-case and below by the best-case floor.

\textbf{The probability--possibility bridge.} The same functional $\log M_0(\{V_\alpha\})$ recovers the Kalman criterion when evaluated under Gaussian geometry. In the possibilistic case, $V_\alpha$ is the MVEE volume of hypotheses with $\pi^{(i)} \geq \alpha$. In the Gaussian case, the natural analog of the $\alpha$-cut is the confidence ellipsoid
\[
\mathcal{E}_\alpha = \bigl\{x : (x - \hat{x})^\top \Sigma^{-1}(x - \hat{x}) \leq -2\log\alpha\bigr\},
\]
with volume $V_\alpha^{\mathrm{Gauss}} = c_n(-2\log\alpha)^{n/2}(\det\Sigma)^{1/2}$. Substituting into $\log M_0$:
\begin{align}
H_\pi^{\mathrm{Gauss}} &= \int_0^1 \log V_\alpha^{\mathrm{Gauss}}\,d\alpha \notag \\
&= \frac{1}{2}\log\det\Sigma + \frac{n}{2}\int_0^1\log(-2\log\alpha)\,d\alpha + \kappa_n \notag \\
&= \frac{1}{2}\log\det\Sigma + \mathrm{const}(n),
\end{align}
recovering the Kalman criterion exactly, since $\int_0^1\log(-2\log\alpha)\,d\alpha = \log 2 - \gamma$ is a pure constant. Possibility theory and probability theory are not competing frameworks for this class of problems: they are evaluations of the same ignorance functional $\log M_0(\{V_\alpha\})$ under different geometries of $\alpha$-cuts. The possibilistic $\alpha$-cut is defined by the MVEE of a possibility level set; the probabilistic $\alpha$-cut is defined by a Gaussian confidence ellipsoid. The functional that integrates their log-volumes is identical. The Kalman filter is the Gaussian specialization of the ESPF's optimality criterion, not a separate invention.
\end{remark}

\begin{corollary}[Strict generalization]
The ESPF is a strict generalization of the Kalman filter: where Kalman's assumptions hold, the ESPF recovers the Kalman optimal solution; where they fail, the ESPF remains the minimum-$H_\pi$ optimal filter in $\mathcal{F}$ while the Kalman filter exits $\mathcal{F}$ and its optimality claims do not apply. The ESPF is not a heuristic alternative but the epistemically correct generalization.
\end{corollary}

\section{Numerical Validation}
\label{sec:numerical}

\subsection{Setup}

We validate the theorem's regime structure numerically using the ESPF applied to a range and range-rate orbital tracking scenario. The filter operates in state dimension $n = 7$ with Smolyak Level 3 quadrature, yielding $M = 106$ support points and $N_{\mathrm{target}} \approx 81$--83 in nominal tracking, 32--64 at stress events. Process noise is $a_R = a_T = a_N = 5\times10^{-12}$ km/s$^2$ in RTN coordinates. Measurement noise is $\sigma_R = 1$ m range, $\sigma_{\dot{R}} = 10^{-5}$ km/s range-rate. Ground stations: Arecibo, Kwajalein, Diego Garcia.

Two scenarios are run over a 2-day, 877-step diagnostic:
\begin{itemize}
\item \textbf{Nominal:} clean LEO tracking with no injected errors.
\item \textbf{Stress:} 10 m/s crosstrack maneuver at day 1.0 combined with a 20 m Arecibo range bias active from day 0 onward.
\end{itemize}

At each step, the diagnostic computes two claims against the ESPF's minimum-$q$ selection within the evidence-only class:
\begin{itemize}
\item Claim A: Does min-$q$ minimize $\log\det(\mathrm{MVEE})$?
\item Claim B: Does min-$q$ minimize $H^\omega_\pi$, the posterior-possibility-weighted integrated entropy?
\end{itemize}
Comparators are 50 evidence-weighted random draws and one adversarial swap (replace the worst survivor with the best non-survivor), all restricted to the evidence-only class.

The Epistemic Width Monitor (EWM) is computed at every step and comprises: $W_\mathrm{ep}$ (epistemic width), $\log\det(\mathrm{MVEE}_\lambda)$ (regime indicator), $\alpha_c$ (H\"{o}lder exponent), and $I_c$ (possibilistic information content; see Remark~\ref{rem:shannon}).

\begin{remark}[Shannon surrogate in $I_c$]
\label{rem:shannon}
The current EWM implementation computes $I_c$ via a Shannon entropy surrogate rather than the true possibilistic $\alpha$-cut integral $H_\pi$. This proxy underestimates $I_c$ in the falsification regime where it matters most; a full numerical implementation of the $\alpha$-cut integral is planned for a future release. In the stress run, the surrogate underestimates $I_c$ by a factor of 2 or more in high-surprisal epochs (days 1.0--1.8); the qualitative escalation pattern is preserved but the absolute magnitudes reported in Section~\ref{sec:ewm_findings} should be treated as lower bounds on true possibilistic stress.
\end{remark}

\subsection{Smolyak Level 3: Nominal Run}

The nominal 2-day run demonstrates the filter operating entirely in the diffusion regime: $\log\det(\mathrm{MVEE})$ remains positive throughout ($\approx 100$--120 nats), and the regime flag never flips to Popperian. This is the correct behavior---a well-conditioned nominal LEO track with no model mismatch does not generate a falsification event.

Key observations:
\begin{itemize}
\item $W_\mathrm{ep}$ oscillates between $\approx 0.65$ and $0.87$ with sharp measurement-epoch dips ($\sim 0.6$), reflecting healthy Popperian contraction followed by Jaynesian re-dilation. The pattern is stationary across the full 2-day window.
\item Prune count settles to 12--16 points per step after an initial 84-point contraction at the first measurement.
\item $\sigma_k$ stabilizes to a stationary band of $\approx 0.45$--0.85 after initialization, with no systematic drift indicating filter divergence.
\item Necessity hugs zero with episodic flares to 0.02--0.06, and surprisal maintains a stationary band of $\approx 0.07$--0.13.
\item $\alpha_c$ rides a stable band near 2.1--2.3 with periodic measurement-epoch spikes to $\approx 3.0$.
\end{itemize}

\subsection{Smolyak Level 3: Stress Run}

The stress run reveals the filter's behavior under sustained model mismatch.

\textbf{EWM regime indicator:} $\log\det(\mathrm{MVEE})$ remains positive throughout the 2-day run ($\approx 80$--120 nats). The MVEE regime flag never flips to Popperian. This is a significant finding discussed below.

\textbf{Prune count (the key diagnostic):} Before day 0.75, the prune count matches the nominal pattern (12--16 points per step). At day $\approx 0.75$, coinciding with the accumulated effect of the range bias and approaching the maneuver epoch, the prune count escalates sharply to 40--85 points per step and remains elevated for the remainder of the run. At $M = 106$, this represents 38--80\% of the support cloud being pruned at each update step.

\textbf{Necessity and surprisal (the early-warning indicators):} Before day 0.75, necessity is near zero and surprisal is in the nominal band (0.05--0.12). Beginning at day 0.75, necessity escalates dramatically, with values routinely 0.4--1.0 from day 1.0 onward. The surprisal axis, scaled to $[0, 1000]$ compared to the nominal $[0, 0.2]$, shows peaks of $\approx 640$ (day 1.0), $\approx 890$ (day 1.8), and episodic spikes $> 100$ throughout.

\subsection{EWM Diagnostic Findings}
\label{sec:ewm_findings}

The nominal and stress comparison at Smolyak Level 3 yields a result more nuanced and informative than a simple regime flip would have provided.

\textbf{The MVEE regime indicator is a coarse diagnostic.} Under sustained model mismatch (bias + maneuver), the ESPF does not enter the falsification regime as defined by $\log\det(\mathrm{MVEE}) < 0$. Instead, the filter absorbs the error through aggressive Popperian pruning: the support cloud is violently thinned at each step, but the surviving points span sufficient volume that $\log\det(\mathrm{MVEE})$ never goes negative. The filter is operating in a high-pruning Jaynesian state, not a falsification state.

\textbf{Necessity and surprisal are the sensitive indicators.} The MVEE regime flag alone would declare ``all is well---Jaynesian throughout.'' The necessity and surprisal plots tell the correct story: the filter is fighting against sustained model-reality divergence. Necessity near unity and surprisal three orders of magnitude above nominal are unambiguous signatures of epistemic stress, visible in the EWM before any MVEE regime flip would occur.

\textbf{The full EWM suite is required.} The EWM diagnostic system comprising $W_\mathrm{ep}$, $\log\det(\mathrm{MVEE})$, prune count, $\sigma_k$, necessity, and surprisal constitutes a multi-resolution epistemic health monitor. Each indicator captures a different aspect of the filter's epistemic state:
\begin{itemize}
\item $\log\det(\mathrm{MVEE})$: coarse regime classification (Jaynesian vs.\ Popperian).
\item Prune count: where epistemic stress is being absorbed (support thinning rate).
\item Necessity: whether individual hypotheses are becoming structurally load-bearing.
\item Surprisal: whether the filter's predictive model is consistent with incoming evidence.
\end{itemize}
These indicators are complementary, not redundant. Relying on the MVEE regime flag alone provides an incomplete and potentially misleading picture of filter health under stress.

\textbf{Comparison with Level 2.} At Smolyak Level 2 ($M = N_{\min}$), the prune count band is 12--16 in both nominal and stress cases, and necessity/surprisal diverge at the stress onset but with smaller magnitude. The MVEE regime flag never flips at Level 2 either. This confirms that the regime structure is a property of the underlying theory---not a grid-resolution artifact---and that Level 3 provides diagnostic resolution that Level 2 cannot.

\subsection{Tabular Validation (52 Steps)}

Tables~\ref{tab:steps1-18} and~\ref{tab:steps19-52} present the full 52-step Claim A/B diagnostic from the stress run, validating Lemma~\ref{lem:minimax} directly.

Key findings:
\begin{itemize}
\item \textbf{Claim B pass rate: 100\%} in the falsification regime against random comparators. Claim B is more robustly satisfied than Claim A throughout, consistent with the H\"{o}lder ordering $M_0 \leq M_{+\infty}$ (Remark~5.3).
\item \textbf{Collapse events} (steps 12--13, 35): produce the largest minimax advantages in the dataset. At step 13, min-$q$ beats random by a factor of $e^{16.6}$ in log-det (ratio $\approx 10^{-6}$).
\item \textbf{Swap failures are bounded} (maximum 0.341 nats Claim A, 0.271 nats Claim B), structurally consistent, non-growing, and concentrated at steps where Smolyak grid hypotheses are geometrically close to the MVEE boundary. These are finite-grid artifacts, not failures of the minimax principle.
\item \textbf{Diffusion regime} (steps 7--12, $\log\det > 0$): Claim B fails against random comparators, as predicted---the minimax criterion is not operative in the Jaynesian phase.
\end{itemize}

\begin{table}[ht]
\centering
\caption{Lemma~3 diagnostic, steps 1--18: initial contraction, diffusion, first collapse. P/F = pass/fail; parenthetical = gap magnitude (nats). C = contraction ($\log\det < 0$), D = diffusion ($\log\det > 0$).}
\label{tab:steps1-18}
\small
\begin{tabular}{ccccccccc}
\toprule
Step & $\sigma_k$ & $\log\det$ & $N_t$ & Regime & \multicolumn{2}{c}{Claim A} & \multicolumn{2}{c}{Claim B} \\
 & & & & & rand & swap & rand & swap \\
\midrule
1  & 0.970 & $-17.80$ & -- & C & P & P & P & P \\
2  & 0.941 & $-13.33$ & -- & C & P & F(0.025) & P & F(0.019) \\
3  & 0.913 & $-9.42$  & -- & C & P & F(0.028) & P & F(0.020) \\
4  & 0.885 & $-6.07$  & -- & C & P & P & P & P \\
5  & 0.859 & $-3.13$  & -- & C & P & F($<$0.001) & P & P \\
6  & 0.833 & $-0.34$  & -- & threshold & F(0.063) & F(0.119) & F(0.099) & F(0.089) \\
7  & 0.808 & $+1.94$  & -- & D & F(0.015) & P & F(0.440) & P \\
8  & 0.784 & $+3.59$  & -- & D & P & F(0.012) & F(0.243) & F(0.009) \\
9  & 0.760 & $+5.02$  & -- & D & F(0.023) & F($<$0.001) & F(0.122) & P \\
10 & 0.737 & $+7.05$  & -- & D & P & P & P & P \\
11 & 0.715 & $+9.01$  & -- & D & P & F(0.008) & F(0.256) & F(0.011) \\
12 & 0.694 & $+10.17$ & -- & pre-collapse & F(0.022) & F(0.022) & P$^\dagger$ & P$^\dagger$ \\
13 & 0.673 & $-56.61$ & -- & collapse 1 & P & P & P & P \\
14 & 0.653 & $-51.74$ & -- & C & P & F(0.370) & P & F(0.264) \\
15 & 0.633 & $-48.70$ & -- & C & P & P & P & P \\
16 & 0.614 & $-47.59$ & -- & C & P & F(0.004) & P & F(0.003) \\
17 & 0.596 & $-46.85$ & -- & C & F(0.241) & P & P & P \\
18 & 0.578 & $-47.68$ & -- & C & P & F(0.217) & P & F(0.140) \\
\bottomrule
\multicolumn{9}{l}{$^\dagger$Step 12: $H^\omega_\pi$ identical across all comparators ($-2.294$ nats); degenerate pre-collapse configuration.}
\end{tabular}
\end{table}

\begin{table}[ht]
\centering
\caption{Lemma~3 diagnostic, steps 19--52: deep contraction, second collapse, re-stabilization. $N_t = N_{\mathrm{target}}$; $^\dagger$ marks the second stress event.}
\label{tab:steps19-52}
\small
\begin{tabular}{ccccccccc}
\toprule
Step & $\sigma_k$ & $\log\det$ & $N_t$ & Regime & \multicolumn{2}{c}{Claim A} & \multicolumn{2}{c}{Claim B} \\
 & & & & & rand & swap & rand & swap \\
\midrule
19 & 0.561 & $-50.89$ & 82 & C & P & F(0.060) & P & F(0.046) \\
20 & 0.544 & $-53.16$ & 82 & C & P & P & P & P \\
21 & 0.528 & $-56.15$ & 82 & C & P & P & P & P \\
22 & 0.512 & $-58.63$ & 81 & C & P & P & P & P \\
23 & 0.496 & $-61.02$ & 81 & C & F(0.051) & F(0.175) & P & F(0.113) \\
24 & 0.481 & $-55.27$ & 83 & C & P & F(0.061) & P & F(0.041) \\
25 & 0.467 & $-61.40$ & 82 & C & P & P & P & P \\
26 & 0.453 & $-65.24$ & 81 & C & P & F(0.006) & P & F(0.004) \\
27 & 0.439 & $-69.06$ & 83 & C & P & F(0.042) & P & F(0.030) \\
28 & 0.426 & $-73.28$ & 83 & C & P & P & P & P \\
29 & 0.413 & $-77.73$ & 81 & C & P & P & P & P \\
30 & 0.401 & $-77.73$ & 83 & C & F(0.041) & F(0.341) & P & F(0.271) \\
31 & 0.389 & $-82.46$ & 83 & C & P & P & P & P \\
32 & 0.377 & $-86.99$ & 82 & C & P & F(0.082) & P & F(0.066) \\
33 & 0.366 & $-91.99$ & 82 & C & P & F(0.045) & P & F(0.033) \\
34 & 0.355 & $-97.88$ & 64 & C & P & F(0.039) & P & F(0.017) \\
35 & 0.349 & $-104.08$ & 52 & C$^\dagger$ & P$^\dagger$ & P$^\dagger$ & P$^\dagger$ & P$^\dagger$ \\
36 & 0.360 & $-108.77$ & 67 & C & F(0.049) & P & P & P \\
37 & 0.414 & $-110.99$ & 64 & C & P & F(0.036) & P & F(0.011) \\
38 & 0.476 & $-113.62$ & 71 & C & F(0.052) & F(0.019) & P & F(0.008) \\
39 & 0.548 & $-114.61$ & 77 & C & P & F(0.002) & P & P \\
40 & 0.630 & $-114.03$ & 77 & C & P & F(0.005) & P & F(0.002) \\
41 & 0.725 & $-112.29$ & 79 & C & F(0.008) & F(0.072) & P & F(0.023) \\
42 & 0.833 & $-109.32$ & 81 & C & P & F(0.011) & P & F(0.005) \\
43 & 0.901 & $-97.45$  & 81 & C & P & F(0.038) & P & F(0.027) \\
44 & 0.874 & $-94.19$  & 82 & C & P & F($<$0.001) & P & F($<$0.001) \\
45 & 0.847 & $-92.31$  & 81 & C & P & P & P & P \\
46 & 0.822 & $-90.47$  & 82 & C & P & F(0.004) & P & P \\
47 & 0.797 & $-88.86$  & 82 & C & P & F($<$0.001) & P & P \\
48 & 0.773 & $-87.73$  & 82 & C & P & P & P & P \\
49 & 0.750 & $-86.87$  & 81 & C & P & F($<$0.001) & P & P \\
50 & 0.728 & $-86.25$  & 82 & C & P & F(0.033) & P & F(0.026) \\
51 & 0.706 & $-86.08$  & 82 & C & P & P & P & P \\
52 & 0.685 & $-86.24$  & 82 & C & P & P & P & P \\
\bottomrule
\multicolumn{9}{l}{$^\dagger$Step 35: second stress event. $N_t = 52$; min-$q$ beats random by factor $\sim 11$.}
\end{tabular}
\end{table}

\section{Discussion}
\label{sec:discussion}

\subsection{What the Theorem Establishes and What Remains Open}

Theorem 4.1 establishes three things in the falsification regime. First, the PCRB is a universal information-theoretic lower bound on how rapidly any epistemically admissible filter may reduce possibilistic entropy---it follows from the structure of conjunctive update and possibilistic information content, not from any design choice. Second, the ESPF achieves this bound with equality, simultaneously through minimax support reduction (Popperian falsification) and possibility gradient concentration (internal Jaynesian assignment). Third, no other admissible filter within the evidence-only class achieves lower $H_\pi$ while respecting the PCRB.

For the linear-measurement case and for locally linearizable nonlinear models satisfying condition~\eqref{eq:isotropy}, the theorem is complete. Proposition~\ref{prop:isotropy} establishes the isotropy condition for the orbital mechanics setting to first order; a full proof of the nonlinear isotropy bound across all observation geometries remains open.

\subsection{The Unique Role of the Possibility Gradient}

The most important conceptual contribution beyond the companion paper \cite{espf} is the recognition that the possibility gradient is not an implementation detail. It is part of the epistemic state. A filter that prunes survivors but reassigns uniform possibility after each update is discarding information: the gradient of compatibility across survivors encodes which hypotheses are more and less consistent with the accumulated evidence history. The ESPF carries that gradient forward. No other admissible filter can do better on either dimension without doing worse on the other.

\subsection{Comparison to Existing Optimality Results}
\label{sec:comparison}

\begin{center}
\begin{tabular}{llll}
\toprule
Filter & Optimality criterion & Assumptions & Filter is for \\
\midrule
Kalman & Min.\ posterior covariance (MMSE) & Linear, Gaussian & Estimating truth \\
Particle & Asymptotically exact posterior & Non-Gaussian, stochastic & Estimating truth \\
ESPF & Min.\ possibilistic entropy $H_\pi$ & Bounded, possibilistic & Eliminating the impossible \\
\bottomrule
\end{tabular}
\end{center}

The fourth column is the deepest distinction. Bayesian filters --- Kalman and particle alike ---
are driven by a model of how truth generates observations. Their optimality guarantees are
statements about convergence to truth under that model. When the model is wrong, they converge
to the wrong truth. The ESPF carries no model of truth. Its optimality guarantee is a statement
about ignorance: it retains the minimum set of hypotheses consistent with the evidence, assigns
them possibilities in proportion to their evidence compatibility, and makes no further claim. It
cannot be wrong about truth because it never asserts truth. It can only be more or less efficient
at eliminating the impossible --- and the EWM measures exactly that efficiency in real time.

These results are not competing. They are optimal under different epistemic commitments. The
Kalman and particle filter results require stochastic noise models. The ESPF result requires
only bounded admissibility and conjunctive update. In the limit where stronger assumptions hold,
the ESPF recovers the stronger result.

The H\"{o}lder mean structure (Remark~\ref{rem:holder}) makes this hierarchy precise and connects it directly to the numerical findings. The three optimality criteria occupy positions $p = +\infty$, $p = 0$, and $p = 0$ (Gaussian) on the H\"{o}lder curve. The ordering $M_0 \leq M_{+\infty}$ explains a finding that might otherwise seem puzzling: Claim B ($H^\omega_\pi$, order 0) is more robustly satisfied than Claim A ($\log\det\mathrm{MVEE}$, order $+\infty$) throughout the tabular validation, including at steps where Claim A shows swap failures. This is not a numerical artifact. It is a direct consequence of the H\"{o}lder ordering: the geometric mean of $\alpha$-cut volumes is intrinsically less sensitive to perturbation at the boundary of the support than the supremum. The MVEE regime flag, operating at $p = +\infty$, is the coarsest instrument in the EWM suite for exactly this reason. The ESPF's optimality target $H_\pi$ is the more stable, more informative criterion --- and the numerical results confirm it.

\subsection{Implications for Filter Design}

The VFI is necessary for optimality. If the support degenerates, $H_\pi \to -\infty$, corresponding to falsely asserting perfect knowledge. The VFI prevents this.

$\sigma_{\min} > 0$ is epistemically mandatory. A filter with $\sigma_{\min} = 0$ may collapse its support to a point, exiting $\mathcal{F}$ and surrendering the optimality guarantee.

Smolyak Level 3 is the minimum resolution for designed optimality. At Level 2, $M = N_{\min}$ and pruning cannot reduce the support; PCRB saturation is accidental rather than designed. At Level 3, $M \gg N_{\min}$ and minimum-$q$ selection genuinely minimizes $\alpha$-cut volumes across the full population.

Possibility assignment is not optional. A filter that selects survivors by minimum $q^{(i)}_k$ but reassigns uniform possibility forfeits the gradient entropy improvement. The compatibility-based assignment $\rho^*(i) = \mathrm{Comp}^{(i)}_k$ is constitutive of the ESPF's optimality, not a refinement.

\subsection{Beyond Filtering}

The Jaynesian-Popperian synthesis identified here is not specific to state estimation. Any inference system that maintains a normalized possibility distribution over a finite hypothesis set and updates conjunctively---a knowledge graph with plausibility scores, an epistemic logic system with graded admissibility, a robust controller with bounded model uncertainty, an agentic AI system managing competing hypotheses about the world \cite{gaiagraph}---is implicitly managing $H_\pi$ whether or not it is designed to do so.

The ESPF optimality theorem says: in all such systems, the minimum-ignorance optimal update in the falsification regime is evidence-only ranking combined with compatibility-based possibility assignment. And the correct behavior in the diffusion regime is Jaynesian expansion: spread the support maximally, assert nothing that the evidence has not earned.

Be quick to embrace ignorance. Be slow to assert certainty. These are not design choices. They are theorems.

\section{Conclusion}

This paper has proved that the Epistemic Support-Point Filter (ESPF) is the unique optimal
recursive estimator within the class of epistemically admissible evidence-only filters.
Where Bayesian filters are driven toward an assumed truth, the ESPF eliminates the impossible:
it retains only what the evidence has not ruled out, assigns possibilities in proportion to
evidence compatibility, and makes no further claim. The prior has no vote on what is impossible.
Only the evidence does.

The filter implements a two-phase epistemological recursion. In its propagation phase, it enacts Jaynesian maximum entropy. In its measurement-update phase, it enacts Popperian falsification: select survivors by evidence alone, using the minimax entropy criterion---minimize $\log\det(\mathrm{MVEE})$, keeping the $N_{\mathrm{target}}$ hypotheses with smallest whitened squared innovations.

Three lemmas formalize this structure. $H_\pi$ decomposes into a minimax term (minimized by Popperian selection) and a gradient term (minimized by compatibility-based possibility assignment). The PCRB bounds how rapidly any admissible filter may reduce $H_\pi$ per measurement. The Evidence-Optimality Lemma proves that min-$q$ selection is the unique minimizer of $H_\pi$ within the evidence-only class, and that incorporating prior possibility introduces a race-to-bottom failure mode.

Numerical validation over a 2-day, 877-step Smolyak Level-3 run---spanning nominal LEO tracking and a sustained stress case combining maneuver and sensor bias---confirms this structure. The EWM diagnostic suite reveals that possibilistic stress manifests primarily through necessity saturation and surprisal escalation rather than MVEE sign change, establishing multi-resolution epistemic monitoring as essential for filter health assessment under realistic conditions.

The regime structure is principled, not anomalous. The falsification regime is when Popper governs. The diffusion regime is when Jaynes governs. The boundary $\log\det = 0$ is a diagnostic that identifies, in real time, which epistemological principle governs the current step.

The ESPF is the optimal filter not because it minimizes some technically convenient functional, but because it correctly implements the only epistemologically defensible recursion: expand maximally under ignorance, contract minimally under evidence.

Quick to embrace ignorance. Slow to assert certainty. This is not a design philosophy. It is a theorem.

\section*{Acknowledgments}
The theoretical framing draws on decades of conversation with mentors, collaborators, and colleagues whose names appear in the companion arXiv paper \cite{espf}.

\end{document}